\documentclass[letterpaper, 10 pt, conference]{ieeeconf}  

\IEEEoverridecommandlockouts                              
\usepackage{cdc-packages}
\usepackage{article-macros-technical}
\usepackage{article-macros-maths}

\title{\LARGE \bf
Error Analysis of Sampling Algorithms \\ for Approximating Stochastic Optimal Control

}

\author{Anant A. Joshi, Amirhossein Taghvaei, Prashant G. Mehta
	\thanks{A.~A.~Joshi and P.~G.~Mehta are with the Coordinated
	Science Laboratory and the Department of Mechanical Science and
	Engineering at the University of Illinois at Urbana-Champaign
	(e-mail: anantaj2; mehtapg@illinois.edu).  A.~Taghvaei is with the
        Department of Aeronautics and Astronautics at the University of
        Washington Seattle (email: amirtag@uw.edu).}         
}

\begin{document}

\maketitle
\thispagestyle{empty}
\pagestyle{empty}

\begin{abstract}

This paper is concerned with the error analysis of two types of sampling algorithms, namely model predictive path integral (MPPI) and an interacting particle system (\IPS) algorithm, that have been proposed in the literature for numerical approximation of the stochastic optimal control.  The analysis is presented through the lens of Gibbs variational principle.  For an illustrative example of a single-stage stochastic optimal control problem, analytical expressions for approximation error and scaling laws, with respect to the state dimension and sample size, are derived. The analytical results are illustrated with numerical simulations.
\end{abstract}

\section{Introduction}

The Gibbs variational problem is given by
\begin{align}\label{eq:KLcost}
\Gibbs{\gamma}{\lag} := &\argmin_{\mu \ll \gamma} \; \left\{ \Div{\mu}{\gamma}  + \mu(c) \right\},
\end{align}
where $\gamma$ is a given probability measure on $\Re^d$ referred to as the {\em prior},  $c:\Re^d \to [0,\infty)$ is a given measurable function that satisfies the integrability constraint $\gamma(c) = \int c(x) \gamma (\ud x) <\infty$, and $\Div{\mu}{\gamma}$ is the relative entropy (or KL divergence) between $\mu$ and $\gamma$.  The minimizer is referred to as the Gibbs measure given by    
\begin{align}\label{eq:optimal-measure}
\mu\opt(\ud x) := \frac{\exp(-\lag(x))\gamma(\ud x)}{\gamma(e^{-c}) }, \quad x\in\Re^d.
\end{align}
Additional technical details and background on this problem and its solution are given in App. \ref{app:gibbs}. 

Although the Gibbs variational problem is classical, there are well-known connections to optimal filtering and optimal control, which have been useful in the development of algorithms.  This paper is concerned with the connection to optimal control, specifically the path-integral control framework pioneered in \cite{kappen-2012,theodorou-2012}.
The framework has led to the development and application of the class of sampling algorithms referred to as the model predictive path integral (MPPI) control \cite{williams-2017}.

MPPI is based on the importance sampling (IS) algorithm. The objective is to numerically approximate the Gibbs measure in terms of samples from the prior. In its basic form, IS proceeds in the following steps:
\begin{enumerate}
\item Sample from the prior $\gamma$
\[
\{X_0^i\}_{i=1}^{N} \iid \gamma.
\]
\item Compute the importance weights
\[
\eta^i := \frac{\exp(-c(X_0^i))}{\sum_{i=1}^{N}\exp(-c(X_0^i))}, \quad 1\leq i\leq N.
\]
\item Obtain an empirical approximation of the Gibbs measure
\[
(\mu\opt)^{(N)} := \sum_{i=1}^{N} \eta^i \delta_{X_0^i}
\]
where $\delta_x$ is the Dirac delta measure at $x\in\Re^d$. 
\end{enumerate}
In the limit as $N \to \infty$, the random measure $(\mu\opt)^{(N)}$ converges weakly in probability to the Gibbs measure $\mu\opt$ \cite[Thm. 1.2]{chatterjee-2018}. 

In the MPPI algorithm, the prior $\gamma$ has the interpretation of the law of a controlled Markov process with a given control (e.g. an open-loop control $u=0$).  The IS algorithm is useful to compute an approximation of the optimal control simply by computing the importance weights.  
This has applications in robotics, 
see for example,
\cite{williams-2016,williams-2018,wan-2021,kurtz-2024}.

\subsection{Aims and original contributions of this paper}
Our goal is to investigate the performance of the MPPI algorithm, as a function of the state dimension $d$ and the number of particles $N$.  For pedagogical reasons, the simplest possible model of a single-stage stochastic optimal control problem (SOCP) is considered as follows:  
\begin{subequations}\label{eq:DTOC1-simple}
\begin{align}
\min_{U \in \U} & \, \E{\half|X_1|^2 + \half |U|^2} \\
\text{s.t.} & \quad X_{1} = x_0 + U + V_1, \quad V_1 \sim \normal(\zero,\id_d)
\end{align}
\end{subequations}
where the state at time $t=0$, $x_0\in\Re^d$, is deterministic, $U \in \U$ is the control input, and $\U$ is the space of admissible control inputs (this is introduced in the main part of the paper).  While the formula for optimal control is elementary, a goal of this paper is to describe its construction and  approximation in the path-integral framework.  The paper makes two original contributions:
\begin{enumerate}
\item Apart from the MPPI algorithm, an alternate, interacting particle system algorithm (\IPS), inspired from ensemble Kalman filter (EnKF), is introduced to solve the same sampling task -- that of approximating the optimal control.  
\item For both the MPPI and the \IPS{}  algorithm, a closed-form formula is presented for the approximation error (in estimating the optimal control) as a function of $N$ and $d$. Analytical results are illustrated with simulations. Simulation results are shown in Fig. \ref{fig:1} and details are provided in Sec. \ref{sec:sim} in the main body of the paper. 
\end{enumerate}
It is hoped that a consistent use of the notation and the algorithms, in the simplest possible settings, will invite convergence between various algorithms whose development has proceeded somewhat independently. 

\begin{figure*}
\centering
\includegraphics[scale=0.72]{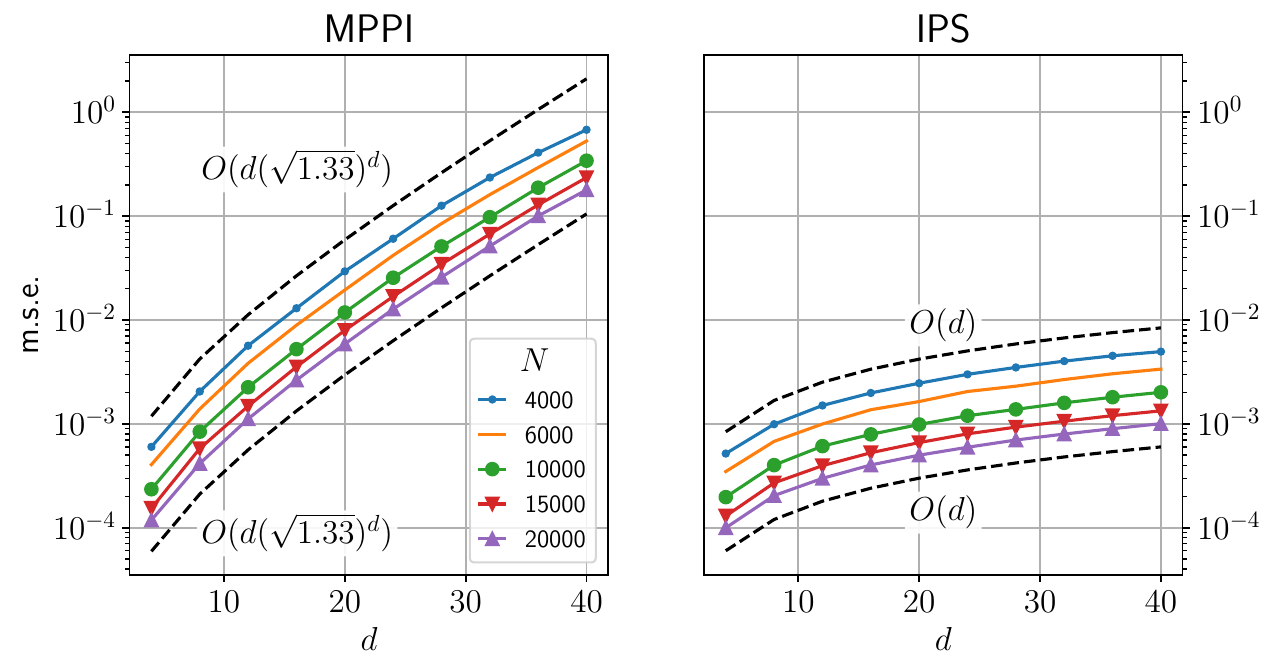}
\caption{Numerical illustration of Prop. \ref{prop:onestep}, with simulation details provided in Sec. \ref{sec:sim}. The y-axis shows the mean square error in estimating the optimal control (m.s.e.)  using the MPPI and IPS algorithm respectively, and the x-axis is the state dimension $d$. The m.s.e. in MPPI grows exponentially in $d$, and m.s.e. in IPS grows affinely in $d$.}
\label{fig:1}
\end{figure*}

\subsection{Related literature}
The use of sampling algorithms to solve optimal control problems is an area of topical interest \cite{theodorou-2012,kappen-2005-jsm,levine-2018,todorov-2008,rawlik-2013,anant-2022,anant-cdc-2024,reich-2024,maoutsa-2020,maoutsa-2022}. 
The idea of using importance sampling to weight an uncontrolled trajectory has been explored in the MPPI literature \cite[Sec. III]{williams-2018}, and also in the reinforcement learning RL literature \cite[Sec. 3]{levine-2018} including in continuous-time settings \cite[Thm. 1]{kappen-2016}.  Related to this paper, the sample complexity analysis of the MPPI algorithm, with respect to the number of samples $N$, appears in \cite{yoon-2022}, \cite{patil-2024}.  Interacting particle algorithms are inspired by the ensemble Kalman filter (EnKF) \cite{katzfuss-2016} and the feedback particle filter \cite{yang-2016} algorithms for filtering.  In contrast to the importance sampling, all the particles in these approaches have equal weights. 
For linear Gaussian filtering, the EnKF is known to avoid the curse of dimensionality exhibited by IS algorithms \cite[Rmk. 3.9]{amir-cips}, which motivates its use in high dimensional applications, e.g. weather prediction \cite{evensen2006}
(see \cite{bishop-2023} for more references).

\subsection{Outline of the paper}
The remainder of this paper is organized as follows: In Sec.~\ref{sec:gibbs}, a path-integral solution of~\eqref{eq:DTOC1-simple} is described in terms of the Gibbs formalism.  Apart from the connection to optimal control, the connection to optimal filtering is also discussed, and duality between the two problems noted.  
In Sec.~\ref{sec:sampling-alg}, the two types of algorithms, namely MPPI and \IPS, are introduced to approximate the optimal control.  The main result, the formulae for approximation errors using the two algorithms, appears in Prop.~\ref{prop:onestep} and illustrated using numerical simulation in~Sec.~\ref{sec:sim}.


\textbf{Notation:}
let $\normal(\cdot,\cdot)$ denote the Gaussian distribution with the mean being the first argument and covariance being the second. 
Let $\id_d$ denote the $d$-dimensional identity matrix, 
$\mu \ll \gamma$ denote $\mu$ is absolutely continuous with respect to $\gamma$ and $\frac{\ud \mu}{\ud \gamma}$ denote the Radon Nikodym derivative.


\section{Relation of $\Gibbs{\gamma}{c}$ to filtering and control}
\label{sec:gibbs}

\subsection{Relation of $\Gibbs{\gamma}{c}$ to optimal control}
\label{sec:control}

Consider SOCP \eqref{eq:DTOC1-simple}. The space of admissible control inputs is denoted by $\U = \UA$, to emphasize the fact that control input is deterministic (it may depend upon $x_0$ but not upon $V_1$).  The lowercase notation is used to denote an arbitrary element $u\in\UA$.    
   
A straight-forward completion of squares argument is used to show that the optimal control is given by,
\[
u\optada = -\frac{x_0}{2}.
\]
Our aim is to connect $u\optada$ to the Gibbs variational problem.  Noting that the solution of the Gibbs problem is a probability measure, we introduce the following notation to denote the measure of the random variable $X_1$:   
\begin{align*}
\; \;\; \rho^u(\cdot) := \P(X_1 \in \cdot \mid U = u), \; u \in \UA,  
\end{align*}
where then $\rho^0$ is the measure of the un-controlled state (with input $U=0$) and $\rho^{u\optada}$ is the measure of the state using optimal control $U=u\optada$.  

Based on these definitions, we have the following result:

\begin{proposition}\label{prop:control-gibbs}
    Consider \eqref{eq:DTOC1-simple} with $c(x):= \half |x|^2$. For $U=u \in \UA$,
    \begin{align*}
        \Div{\rho^u}{\rho^0} &= \half|u|^2 \\
        \rho^u(c) &= \half\E{|X_1|^2}.
    \end{align*}
\end{proposition}
\medskip
\begin{proof}
    In App. \ref{app:control-1}.
\end{proof}

Consequently, the SOCP~\eqref{eq:DTOC1-simple} is equivalently expressed as follows:
\begin{align}\label{eq:KL-control}
    \argmin_{u \in \UA} \; \left\{ \Div{\rho^u}{\rho^0}  + \rho^u(c) \right\}.
\end{align}
Now, set
\begin{align}\label{eq:Gibbs-control}
    \rho\optnon := \Gibbs{\rho^0}{c}.
\end{align}
One might conjecture that
\[
\rho\optnon  \stackrel{(??)}{=} \rho^{u\optada}.
\]
In Appendix~\ref{app:control-2}, it is shown that the conjecture is false. The control input that yields the Gibbs measure is in fact non-deterministic given by 
\begin{align*}
    U\optnon &:= u\optada + V_1(\frac{1}{\sqrt{2}} - 1),
\end{align*}
such that
\begin{align}\label{eq:gibbs-measure-control}
        \rho\optnon(\cdot) &= \P(X_1 \in \cdot \mid U = U\optnon). 
\end{align}  
In Table \ref{tb:duality-control}, the parallels between the SOCP and the Gibbs problem are tabulated.  The difference between the solution of the two problems is described using the Venn diagram in Fig \ref{fig:gibbs}.

The MPPI algorithm is based on a key result described in the following proposition. 

\begin{proposition}\label{prop:control}
Consider SOCP \eqref{eq:DTOC1-simple}.  Then the optimal deterministic control input is related to the Gibbs measure as follows: \begin{align}\label{eq:optimal-control}
        u\optada = \int_{\R^d}x_1 \rho\optnon(\ud x_1) - x_0.
    \end{align}
\end{proposition}
\vspace{5pt}
\begin{proof}
In App. \ref{app:control-3}.
\end{proof}

\begingroup
\renewcommand{\arraystretch}{1.3}
\begin{table}
\centering
\begin{tabular}{ccc}
\toprule
\textbf{Gibbs, Eq. \eqref{eq:KLcost}} & \multicolumn{2}{c}{\textbf{SOCP, Eq. \eqref{eq:DTOC1-simple}}} \\
Quantity & Quantity & Meaning  \\ \midrule
$\gamma$ & $\rho^0$ & Uncontrolled measure \\
$c(x)$ & $\half |x|^2$ & State cost  \\
$\mu\opt$ & $\rho\optnon$ & Gibbs measure  \\
\bottomrule
\end{tabular}
\caption{Equivalence between \eqref{eq:KLcost} and \eqref{eq:DTOC1-simple}.}
\label{tb:duality-control}
\end{table}
\endgroup

    
    

\begin{figure}[h]
    \centering
    \includegraphics[scale=0.3]{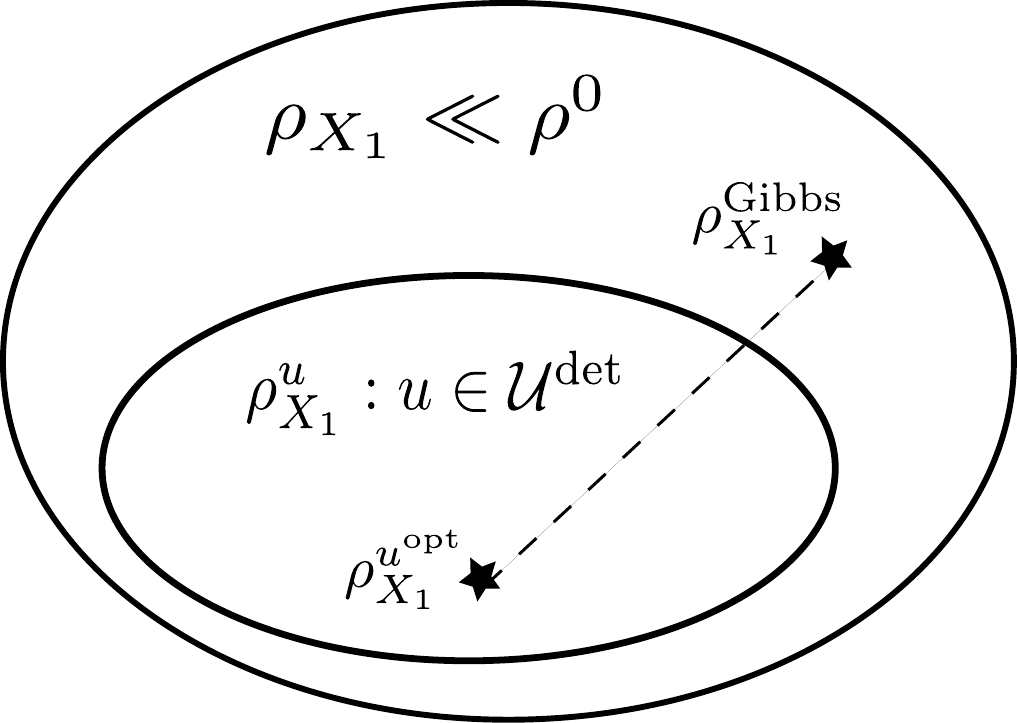}
    \caption{Pictorial description of \eqref{eq:gibbs-measure-control}}
    \label{fig:gibbs}
\end{figure}

\begin{remark}\label{rmk:control-gibbs}
Another approach to compute the optimal control is to solve a variational problem (see \cite[(16)]{williams-2018}),
    \begin{align*}
        u\optada = \argmin_{u \in \UA}\Div{\rho\optnon}{\rho^u},
    \end{align*}
    The calculation of the same leading to the formula \eqref{eq:optimal-control} is also given in App. \ref{app:control-3}.
\end{remark}

\begin{remark}
    To many control theorists, the appearance of non-adapted control may strike as odd. The important point to note is that the formula for the adapted (deterministic) optimal control $u\optada$ is given in terms of $\rho\optnon$ as in \eqref{eq:optimal-control}.  Moreover, the IS procedure is used to approximate the Gibbs measure in terms of samples from $\rho^0$ (which is the uncontrolled state). This remarkable idea was described in the pioneering works \cite{kappen-2012,theodorou-2012}.
\end{remark}

\begin{remark}
    A limitation of the path integral control approach is that the control and noise need to act in the same channel, see \cite[Sec. II]{williams-2018}. A more general form of the dynamics in \eqref{eq:DTOC1-simple} is
    \begin{align*}
        X_1 = x_0 + B(U + V_1), \quad V_1 \sim \normal(\zero,\id_d).
    \end{align*}
    We work with the case $B=\id_d$ in order to describe the main ideas clearly without additional notational burden.
\end{remark}


\begin{remark}
    In this work, we consider measure on the path space $\P(X_1 \in \cdot \, \mid U)$. Equivalently, one may study the SOCP \eqref{eq:DTOC1-simple} using measure on the noise space, that is, $\P(V_1 \in \cdot \, \mid U)$, see for example, \cite{williams-2018}. For \eqref{eq:DTOC1-simple}, both approaches yield the same result.
\end{remark}


\subsection{Relation of $\Gibbs{\gamma}{c}$ to filtering, and duality}
\label{sec:filter}
The dual counterpart of the single stage SCOP \eqref{eq:DTOC1-simple} is the single stage filtering model
\begin{align}\label{eq:DTF}
Z_1 = h(Y_0) + W_1, \quad W_1 \sim\normal(0_m,\RF)
\end{align}
where $h : \R^d \to \R^m$ is a Borel measurable function
and $Y_0$ is assumed to be independent of $W_1$. The random variable $Y_0$ represents the hidden state of a system,  $Z_1$ denotes the observation, and $W_1$ denotes the observation noise.  

The goal of the filtering problem is to find the conditional measure of $Y_0$ given $Z_1$, known as the \textit{posterior}.

By defining $\gamma(\cdot)$ to be the prior and the function $c(\cdot)$ to be log likelihood, as 
\begin{subequations}\label{eq:filter-qty}
    \begin{flalign}
\text{(Prior)} & \quad  \gamma(\cdot) := \P(Y_0 \in \cdot),& & \\
\text{(Log likelihood)} & \quad  \lag(y;z) := \half|z -  h(y)|^2_{\RF\inv}, \; y \in \R^d,& & 
\end{flalign}
\end{subequations}
we arrive at a Gibbs variational formulation of the posterior, presented in the following result. The proof appears in App. \ref{app:filter}.
\begin{proposition}\label{prop:filter}
Consider the filtering model~\eqref{eq:DTF} and define $\gamma(\cdot)$ and $c(\cdot)$ as~\eqref{eq:filter-qty}. Then 
we have the equivalence,
\begin{align*}
\P(Y_0 \in \cdot \mid Z=z) = \Gibbs{\gamma}{c(\cdot;z)}.
\end{align*}
\end{proposition}

\medskip



\begingroup
\renewcommand{\arraystretch}{1.3}
\begin{table}
\centering
\begin{tabular}{ccc}
\toprule
\textbf{Gibbs, Eq. \eqref{eq:KLcost}} & \multicolumn{2}{c}{\textbf{Filter Eq. \eqref{eq:DTF}}} \\ 
Quantity & Quantity & Meaning  \\ \midrule
$\gamma$ & $\P(X_0 \in \cdot)$ & Prior \\
$c(x;z)$ & $\half |z - x|_{R\inv}^2$ & Log likelihood  \\
$\mu\opt$ & $\post(\cdot \mid z)$ & Posterior  \\
\bottomrule
\end{tabular}
\caption{Equivalence between \eqref{eq:KLcost} and \eqref{eq:DTF}.}
\label{tb:duality-filter}
\end{table}
\endgroup

\noindent\textbf{Duality:}
The utility of the filtering model \eqref{eq:DTF} in solving SOCP \eqref{eq:DTOC1-simple} is given by the following duality result.

\begin{proposition}\label{prop:duality}
    Consider \eqref{eq:DTF}, with $Y_0 \sim \normal(x_0,\id_d)$, and $h(x) := x$,
    that is, 
     \begin{align}\label{eq:dual-filter}
        Z_1 &= Y_0 + W_1, \; Y_0 \sim \normal(x_0,\id_d), \; W_1 \sim \normal(\zero,\id_d). 
    \end{align}
    Then the posterior of \eqref{eq:dual-filter} denoted by $\post$ satisfies 
    \begin{align*}
        \post(\cdot \mid \zero) = \rho\optnon(\cdot).
    \end{align*}
    where $\rho\optnon(\cdot)$ is defined in~\eqref{eq:Gibbs-control}.
\end{proposition}
\begin{proof}
    See App. \ref{app:duality}.
\end{proof}

To facilitate clarity for the reader, a summary of the relationship between filtering and Gibbs variational formulation, and the duality, are presented in Table \ref{tb:duality-filter} and \ref{tb:duality}, respectively. 

\begin{remark}\label{rmk:duailty}
    The optimal control can be computed from the posterior of \eqref{eq:dual-filter} using Prop. \ref{prop:control} and \ref{prop:duality} as
    \begin{align*}
        u\optada = \int_{\R^d} \, x\post(\ud x \mid \zero) \, - x_0.
    \end{align*}
\end{remark}

\begingroup
\renewcommand{\arraystretch}{1.3}
\begin{table*}
\centering
\begin{tabular}{cp{1.9cm}p{1.6cm}ccc}
\toprule
\textbf{Gibbs, Eq. \eqref{eq:KLcost}} & \multicolumn{2}{c}{\textbf{Meaning}} &  \multicolumn{3}{c}{\textbf{Quantity}}  \\ 
 & SOCP, Eq. \eqref{eq:DTOC1-simple} & Filter, Eq. \eqref{eq:DTF} &SOCP, Eq. \eqref{eq:DTOC1-simple} & Filter, Eq. \eqref{eq:DTF}  & Duality\\ \midrule
$\gamma$ & Uncontrolled measure & Prior & $\exp({-\half\left|x - x_0 \right|^2})$ & $\exp(-\half\left|x - \theta\right|_{\Theta}^2)$ & $\theta = x_0$, $\Theta = \id_d$ \\ \midrule
$c(x)$ & State cost & Log likelihood & $\half |x|^2$ & $\half |z - x|_{R\inv}^2$ & $z = \zero$, $R = \id_d$ \\ \midrule
{$\mu\opt$} & {Gibbs measure} & {Posterior}  & {$\rho\optnon$}  & {$\post(\cdot \mid \zero)$}  & $\exp(-\left|x - \frac{x_0}{2}\right|^2)$  \\
\bottomrule
\end{tabular}
\caption{Duality between \eqref{eq:DTOC1-simple} and \eqref{eq:DTF} through which the control problem can be posed as a filtering problem and solved using sampling techniques.}
\label{tb:duality}
\end{table*}
\endgroup

\section{Sampling algorithms to find optimal control}
\label{sec:sampling-alg}


\subsection{MPPI}
In Prop. \ref{prop:control}, computing the optimal control
involves an expectation with $\rho\optnon$ which may not be possible, so we use importance sampling with a  known control $\bar{u} \in \UA$, to sample from $\rho^{\bar{u}}(\cdot) := \P(X_1 \in \cdot \mid U = \bar{u})$:
\begin{align}
&\int_{\R^d} x \, \ud\rho\optnon(x) = \int_{\R^d} \left( x \, \frac{\ud \rho\optnon}{\ud \rho^{\bar{u}} }(x) \right) \ud \rho^{\bar{u}}(x) \nonumber \\ 
\noalign{\vskip5pt}
&\qquad \qquad \qquad \qquad = \int_{\R^d} x\, \eta(x;\bar{u})  \, \ud \rho^{\bar{u}}(x) \label{eq:iscontrol}
\end{align}
where the importance sampling weight is
\begin{align}\label{eq:eta}
 &\eta(x;\bar{u}) := \frac{\exp\left( -\half \left| x + \bar{u}\right|^2 \right) }{\int_{\R^d}{\exp\left( -\half \left| x' + \bar{u}\right|^2 \right)} \, \ud \rho^{\bar{u}}(x')}.
\end{align}

This naturally leads us to the importance sampling based particle algorithm \cite[Alg. 1]{williams-2018} as follows:
\begin{enumerate}
    \item Sample from prior $\gamma$:
    \begin{align*}
    &\{V_0^i\}_{i=1}^{N} \iid \normal(\zero,\id_d) \\
    &X_1^i := x_0 + \bar{u} + V_0^i,\quad 1 \le i \le N, 
    \end{align*}
    which gives that $\{X_1^i\}_{i=1}^{N} \iid \rho^u$.
    \item Construct importance weights:
    \begin{align*}
        \eta^i := \frac{\exp\left( -\half \left| X_1^i + \bar{u}\right|^2 \right) }{\frac{1}{N}\sum_{i=1}^N \exp\left( -\half \left| X_1^i + \bar{u}\right|^2 \right) }, \; \; 1 \le i \le N.
    \end{align*}
    \item Find empirical approximation of optimal control
    \end{enumerate}
%
     \boxedeq{eq:oc-mppi}{
        & u\optada \approx (u\optada_{\text{MPPI}})^{(N)} := \frac{1}{N}\sum_{i=1}^{N} \eta^i X_1^i - x_0.
        }

See 
App. \ref{app:is} for calculation of the change of measure. 
\subsection{Interacting particle system (\IPS)}
\label{sec:enkf}
To solve the filtering problem \eqref{eq:dual-filter}, consider the following mean-field system inspired from the EnKF:
\begin{subequations}\label{eq:meanfield}
    \begin{align}
    &\bar{\Y}_{1} = \bar{\Y}_0 - \bar{L}_0(\bar{\Y}_0 + \bar{W}_0), \; \\
    &\bar{\Y}_0 \stackrel{d}{=} Y_0, \quad  \bar{W}_0 \stackrel{d}{=} W_0
\end{align}
\end{subequations}
where $\bar{\Y}_0$ and $\bar{W}_0$ are independent, and the gain is
\begin{align}\label{eq:meanfieldgain}
    \bar{L}_0 = \cov{\bar{Y}_0}(\cov{\bar{Y}_0} + \cov{\bar{W}_0}).
\end{align}

\begin{proposition}\label{prop:meanfield}    
For the mean-field system \eqref{eq:meanfield} with gain \eqref{eq:meanfieldgain}, 
\begin{align*}
    \P(\bar{\Y}_{1} \in \cdot) = \rho\optnon(\cdot).
\end{align*}
\end{proposition}
\begin{proof}
    See App. \ref{app:meanfield}.
\end{proof}

To implement the mean field system
we use a \IPS:
\begin{align*}
\{ (\Y_0^i, \Y_1^i) \in \R^d \times \R^d :  1 \le i \le N \}.
\end{align*}

The equations for the \IPS:
\begin{gather*}
\Y_{1}^i = \Y_0^i - L_0^{(N)}(\Y_0^i + W_0^i), \; \\
\{\Y_0^i\}_{i=1}^{N} \iid\normal(x_0,\id_d), \quad  \{W_0^i\}_{i=1}^{N} \iid\normal(\zero,\id_d),
\end{gather*}
with the empirical approximation for the gain
\begin{align*}
&L_0^{(N)} := \Sigma^{(N)}_0(\Sigma^{(N)}_0 + \id_d)\inv 
\end{align*}
and $\hat{\Y}_0^{(N)} := \frac{1}{N}\sum_{i=1}^{N}\Y_0^i$,
\begin{align*}
&\Sigma^{(N)}_0 := \frac{1}{N-1}\sum_{i=1}^{N}(\Y_0^i - \hat{\Y}_0^{(N)})(\Y_0^i - \hat{\Y}_0^{(N)})\tp 
\end{align*}
are respectively the empirical mean and covariance.
Then we use Rmk. \ref{rmk:duailty} to approximate the optimal control
\boxedeq{eq:oc-ips}{
 u\optada \approx (u\optada_{\text{\IPS}})^{(N)} := \frac{1}{N}\sum_{i=1}^{N} \Y_1^i - x_0.
}


\subsection{Scaling of error with dimensions}
\label{sec:comparison}

In the following proposition, we study how the importance sampling and interacting particle system methods behave with increasing dimension of state $d$. 

Define the mean square error in estimating the optimal control
\begin{align}\label{eq:mse}
    \mse_{*} :=  \E{|(u\optada_{*})^{(N)}-u\optada|^2},
\end{align}
where $*$ is MPPI or \IPS. 
\begin{proposition}\label{prop:onestep}
Consider the SOCP \eqref{eq:DTOC1-simple}, and the MPPI method \eqref{eq:oc-mppi} and \IPS{} method \eqref{eq:oc-ips} to approximate the optimal control. The mean squared error defined in \eqref{eq:mse}, is given as
\begin{enumerate}
\item For the MPPI algorithm:
\begin{align*}
&\mse_{\text{MPPI}} = \frac{1}{N}\left(\sqrt{\frac{4}{3}} \right)^{d} \left( \left|\frac{\bar{u} - x_0}{3}\ \right|^2 + \frac{d}{3} \right) \\
& \; \exp\left( \frac{4|\bar{u}|^2 + 7 |x_0|^2 + |x_0 + u|^2}{6} \right)  -  \frac{|x_0|^2}{4N} 
\end{align*}
\item For interacting particle system (\IPS) algorithm:
\begin{align*}
\mse_{\text{\IPS}} \le \frac{2d}{N} + \frac{5}{4}|x_0|^2
\end{align*}
\end{enumerate}
\end{proposition}
\vspace{5pt}
\begin{proof}
See App. \ref{app:onestep}.
\end{proof}

\section{Numerical simulation}
\label{sec:sim}

In the numerical simulation, we apprximate the optimal control for \eqref{eq:DTOC1-simple} using the MPPI method \eqref{eq:oc-mppi} and \IPS{} method \eqref{eq:oc-ips}.
In figure \ref{fig:1} we display results of numerical simulations demonstrating the trends in Prop. \ref{prop:onestep}, for $x_0 = 0$. The MPPI algorithm samples from the uncontrolled distribution, that is, $\bar{u} = 0$. 

It is observed that the m.s.e. in MPPI algorithm grows as $O(d(1.33)^{\frac{d}{2}})$ and in \IPS{} algorithm grows as $O(d)$, which is both consistent with Prop. \ref{prop:onestep}. Thus we can see the the curse of dimensionality in the MPPI algorithm, which is avoided by the \IPS{} algorithm.

The scaling in $d$ is investigated for various values of $N$, namely $N=4\times 10^3, 6\times 10^3, 10^4,1.5 \times 10^4, 2\times 10^4$. To calculate the m.s.e. the results are averaged over 1000 independent runs for calculating the expectation. See algorithm \ref{alg:is} and \ref{alg:cips}  for implementation details of $u\optada_{\text{MPPI}}$ and $u\optada_{\text{\IPS}}$ respectively.

\begin{algorithm}
\caption{MPPI algorithm for control \cite[Alg. 1]{williams-2018}}
\label{alg:is}
\begin{algorithmic}[1]
\STATE Sample $\{V_0^i\}_{i=1}^{N} \iid \normal(\zero,\id_d)$
\FOR{$i=1,2,\ldots,N$}
\STATE $X_1^i := x_0 + \bar{u} + V_0^i$
\STATE $\tilde{\eta}^i :=\exp\left( -\half \left| X_1^i + \bar{u}\right|^2 \right)$
\ENDFOR
\FOR{$i=1,2,\ldots,N$}
\STATE ${\eta}^i :=\eta^i(\frac{1}{N}\sum_{i=1}^N \tilde{\eta}^i)\inv$
\ENDFOR
\RETURN $(u\optada_{\text{MPPI}})^{(N)} :=\frac{1}{N}\sum_{i=1}^{N}\eta^iX_1^i - x_0$
\end{algorithmic}
\end{algorithm}

\begingroup
\begin{algorithm}
\caption{Interacting particle algorithm for control}
\label{alg:cips}
\begin{algorithmic}[1]
\renewcommand{\arraystretch}{5}
\STATE Sample $\{\Y_0^i\}_{i=1}^{N} \iid\normal(x_0,\id_d)$, \vspace{3pt}
\STATE $\hat{\Y}_0^{(N)} := \frac{1}{N}\sum_{i=1}^{N}\Y_0^i$, \vspace{3pt}
\STATE $\Sigma^{(N)}_0 := \frac{1}{N-1}\sum_{i=1}^{N}(\Y_0^i - \hat{\Y}_0^{(N)})(\Y_0^i - \hat{\Y}_0^{(N)})\tp$ \vspace{3pt}
\STATE $L_0^{(N)} := \Sigma^{(N)}_0(\Sigma^{(N)}_0 + \id_d)\inv$ \vspace{3pt}
\STATE  $\{W_0^i\}_{i=1}^{N} \iid \normal(\zero,\id_d)$ \vspace{3pt}
\FOR{$i=1,2,\ldots,N$} \vspace{3pt}
\STATE $\Y_{1}^i = \Y_0^i - L_0^{(N)}(\Y_0^i + W_0^i)$, 
\ENDFOR
\RETURN $(u\optada_{\text{\IPS}})^{(N)} := \frac{1}{N}\sum_{i=1}^{N}\Y_1^i - x_0$
\end{algorithmic}
\end{algorithm}
\endgroup


\section{Conclusion and future work}
\label{sec:conc}
In this paper, we revisited the formulation of control problems as the Gibbs variational problem, in order to study how performance of two numerical algorithms namely, MPPI and interacting particle systems (IPS) scale with dimension of state. As future work, we will look at the case of multistep optimal contol with generic coefficients, and nonlinear state cost and state dynamics.

\bibliography{references}
\bibliographystyle{IEEEtran}

\appendices

\section{Gibbs variational problem}
\label{app:gibbs}
Define 
\begin{align*}
\F := -\log \int_{\R^d}\exp(-\lag(x)) \gamma(\ud x) = -\log(\gamma(e^{-c})).
\end{align*}
\begin{lemma}
$\F$ is the minimum of \eqref{eq:KLcost} and the optimal measure which achieves this is \eqref{eq:optimal-measure} given by \eqref{eq:optimal-measure} recalled below:
\begin{align*}
\mu\opt(\ud x) := \frac{\exp(-\lag(x))\gamma(\ud x)}{\gamma(e^{-c}) }, \quad x\in\Re^d.
\end{align*}
\end{lemma}
\vspace{10pt}
\begin{proof}
Let $\mu \ll \gamma$ be arbitrary. Then using change of measure and Jensen inequality,
\begin{align}
\F = -\log \int_{\R^d}\exp(-\lag(x)) \frac{\ud \gamma}{\ud \mu}(x) \, \mu(\ud x) \nonumber \\
\le \int_{\R^d}\left(\lag(x) + \log \frac{\ud \mu}{\ud \gamma}(x) \right)\, \mu(\ud x) \nonumber
\end{align}
Hence $\F$ is a lower bound for \eqref{eq:KLcost} and substituting the formula \eqref{eq:optimal-measure} for $\mu\opt$ produces an equality. 
\end{proof}
\begin{remark}
In statistical mechanics, $\F$ is known as the free energy.
\end{remark}

\section{Details of Sec. \ref{sec:control}}
\label{app:control}
For $u \in \UA$, $\rho^{u} =\normal(x_0 + u,\id_d)$, that is,
\begin{align*}
\rho^0(\ud x) &= \frac{\ud x}{(\sqrt{2\pi})^d}\exp\left(-\half\left|x-x_0 \right|^2\right),  \\
\rho^{u}(\ud x) &= \frac{\ud x}{(\sqrt{2\pi})^d}\exp\left(-\half\left|x - x_0 - u\right|^2\right).
\end{align*}
\subsection{Proof of Prop. \ref{prop:control-gibbs}}
\label{app:control-1}
To get the equivalence between \eqref{eq:KLcost} and \eqref{eq:DTOC1-simple} 
the following result is helpful.
\begin{lemma}
For any $u \in \UA$,
\begin{align*}
&\log \frac{\ud \rho^{u}}{\ud \rho^0}(X_1) = V_1\tp u + \half |u|^2, \\
&\Div{\rho^{u}}{\rho^0} = \rho^u\left(\log \frac{\ud \rho^{u}}{\ud \gamma}(X_1) \right) = \half |u|^2.
\end{align*}
\end{lemma}
\begin{proof}
Using the formula for $\rho^u$ and $\rho^0$ we get 
\begin{align*}
&\log \frac{\ud \rho^{u}}{\ud \rho^0}(x) = (x - x_0)\tp u - \half |u|^2, 
\end{align*}
which when evaluated at $X_1 = x_0 + u + V_1$ gives the first equation. Since $V_1$ is independent of $u$ and has zero mean under $\rho^u$, it leads to the second equation. 
\end{proof}
Recall that $c(x) = \half |x|^2$. Then
\begin{align*}
    \rho^u(c) = \half\int|x|^2 \P(X_1 \in \ud x \mid U = u) = \half \E{|X_1|^2}.
\end{align*}
Hence,
\begin{align*}
\Div{\rho^{u}}{\rho^0} + \rho^{u}(\lag) &= \E{\half|X_1|^2} + \half |u|^2
\end{align*}
which establishes the equivalence between \eqref{eq:KLcost} and \eqref{eq:DTOC1-simple}.

%
\subsection{Details of \eqref{eq:gibbs-measure-control}}
\label{app:control-2}
Calculate the $\rho\optnon$ using \eqref{eq:optimal-measure}:
\begin{align*}
\rho\optnon(\ud x) &= \frac{\exp(-\lag(x))\rho^0(\ud x)}{\int_{\R^d}\exp(-\lag(x')) \rho^0(\ud x') } \\
&= \frac{\exp\left(-\half|x|^2 -\half\left|x-x_0 \right|^2\right)\ud x }{\int_{\R^d}\exp(-\half|x'|^2 -\half\left|x'-x_0 \right|^2)\ud x' } \\
&=  \frac{\ud x}{(\sqrt{\pi})^d}\exp\left(-\left|x - \frac{x_0}{2}\right|^2\right).
\end{align*}
    Thus $\rho\optnon = \normal(\frac{x_0}{2},\frac{1}{2}\id_d)$. Under $U\optnon$, $X_1 = \half x_0 + \frac{1}{\sqrt{2}}V_1$ hence $X_1 \sim \normal(\frac{x_0}{2},\frac{1}{2}\id_d)$.

Now it is clear that $\rho\optnon$ can never be equal to $\rho^u$ since the second moment of $\rho\optnon$ is different from that of$\rho^u$. As seen from $\rho^{u}$, the control can only adjust the mean of the distribution but not the covariance, since it is adapted to $x_0$. 

\subsection{Proof of Prop. \ref{prop:control} and Rmk. \ref{rmk:control-gibbs}}
\label{app:control-3}
To find the optimal control from $\rho\optnon$ we use the following result \cite{williams-2018}.
\begin{lemma}
The following is true for  \eqref{eq:DTOC1-simple}: 
\begin{enumerate}
\item $u\optada = \argmin_{u}\Div{\rho\optnon}{\rho^{u}}$ \vspace{3pt}
\item $u\optada = \int x \, \ud \rho\optnon(x) - x_0 $
\end{enumerate}
\end{lemma}
\begin{proof}
Observe that
\begin{align*}
\Div{\rho\optnon}{\rho^{u}} = \left| x_0 + u - \frac{x_0}{2} \right|^2 + \{\text{const. w.r.t $u$}\}
\end{align*}
hence $u\optada = -\frac{x_0}{2} = \argmin_{u}\Div{\rho\optnon}{\rho^{u}}$. The part (2) can be seen from the expression of $\rho\optnon$.
\end{proof}

\section{Proof of Prop. \ref{prop:filter}}
\label{app:filter}
Let $F_{Y_0|Z_1}(\cdot \mid z) := \P(X_0 \in \ud x\mid Z_1=z)$.
Using Bayes rule we have that
\begin{align*}
F_{Y_0|Z_1}(\cdot \mid z) = \frac{\exp\left(-\half |z - h(x)|^2_{\RF\inv}\right)\gamma(\ud x)}{\int_{\R^d}\exp\left(-\half |z - h(x')|^2_{R\inv}\right)\gamma(\ud x')}
\end{align*}
Recalling the negative of log likelihood function $c(y;z) = \half|z -  h(y)|^2_{\RF\inv}$,
we see that the measure $\mu\opt$ in \eqref{eq:optimal-measure} for $\Gibbs{\gamma}{c(\cdot;z)}$ is the same as the measure $F_{Y_0|Z_1}(\cdot \mid z)$.
Therefore, we have the equivalence,
\begin{align*}
F_{X_0|Z_1}(\cdot \mid z)(\cdot,z) = \Gibbs{\gamma}{c(\cdot;z)}.
\end{align*}

\section{Proof of Prop. \ref{prop:duality}}
\label{app:duality}
Using Bayes' rule,
\begin{align*}
\post(\ud x) &= \frac{\exp\left(-\half|x|^2 -\half\left|x-x_0 \right|^2\right)\ud x }{\int_{\R^d}\exp(-\half|x'|^2 -\half\left|x'-x_0 \right|^2)\ud x' } \\
&=  \frac{\ud x}{(\sqrt{\pi})^d}\exp\left(-\left|x - \frac{x_0}{2}\right|^2\right) = \rho\optnon(\ud x).
\end{align*}

\section{Change of measure in MPPI}
\label{app:is}

\begin{lemma}
Let $\eta$ be as in \eqref{eq:eta}. Then \eqref{eq:iscontrol} holds.
\end{lemma}

\hspace{12pt}\textit{Proof:}
\begin{align*}
&\int_{\R^d} x \, \ud\rho\optnon(x)= \int_{\R^d} x \left(\frac{\ud \rho\optnon}{\ud \gamma} \frac{\ud \gamma}{\ud \rho^{\bar{u}}} \right) \, \ud \rho^{\bar{u}}(x).  
\end{align*}
Now we use the following expressions to get the result:
\begin{align*}
&\frac{\ud \rho\optnon}{\ud \gamma}(x) = \frac{\exp\left(-\half |x|^2 \right)}{\int{\exp\left(-\half |x'|^2 \right)}\gamma(\ud x')}, \\ 
&\frac{\ud\gamma}{\ud \rho^{\bar{u}}}(x) = \frac{\exp\left(\half|\bar{u}|^2 - (x - x_0)\tp\bar{u}\right)}{\int_{\R^d}{\exp\left(\half|\bar{u}|^2 - (x' - x_0)\tp\bar{u}\right)}\rho^{\bar{u}}(\ud x')}. \; \; \rule{0.6em}{0.6em}
\end{align*}

\section{Proof of Prop. \ref{prop:meanfield}}
\label{app:meanfield}
Clearly, since $\cov{\bar{\Y}_0} = \cov{\bar{W}}_0 = \id_d$ we have $\bar{L}_0 = \half\id_d$. 
Hence $\bar{\Y}_1 = \half\bar{\Y}_0 + \half\bar{W}_0$. Since $\bar{\Y}_0$ and $\bar{W}_0$ are independent,  
$\bar{\Y}_{1}\sim \normal(\half x_0,\half\id_d)$. Moreover, from proof of Prop. \ref{prop:duality}, $\post = \rho\optnon = \normal(\half x_0,\half\id_d)$.

\section{Proof of Prop. \ref{prop:onestep}}
\label{app:onestep}
Let  $\E[\xi \sim \mu]{f(\xi)} := \int f(x) \ud\mu(x)$ denote expectation.
To aid analysis, inspired from \cite{amir-cips}, we make the approximation,
\begin{align}
\label{eq:approximation}
\eta^i := \frac{\exp\left( -\half \left| X_1^i + \bar{u}\right|^2 \right)}{\E[X_1^i \sim \rho^{\bar{u}}]{\exp\left( -\half \left| X_1^i + \bar{u}\right|^2 \right)}}.
\end{align}
\noindent\textbf{MPPI:}
We define the quantities 
\begin{align*}
\hat{X} &:= \frac{1}{N}\sum_{i=1}^{N} \eta^i X_1^i, \quad X_1^i \stackrel{\mathrm{i.i.d}}{\sim} \normal(x_0 + \bar{u},\id_d) \\
\eta^i &:= \frac{\tilde{\eta}^i}{\frac{1}{N}\sum_{i=1}^N \tilde{\eta}^i}, \quad
\tilde{\eta}^i := \exp\left( -\half \left| X_1^i + \bar{u}\right|^2 \right) 
\end{align*}
Then the optimal control is $(u\optada)^{(N)}  := \hat{X}_1 - x_0$ and $\E{(u\optada)^{(N)}} = u\optada = -\frac{x_0}{2}$ since $\E{\hat{X}_1} = \frac{x_0}{2}$. Hence,
\begin{align*}
\var{u_0^{(N)}} = \frac{1}{N}\left( \E{|\eta^i X_1^i|^2 - |u\optada|^2} \right).
\end{align*}
Recall that $\bar{u}=0$. Then, the calculations follow as: 
\begin{align*}
r_1 &:= \E[X_1^i \sim \normal(x_0 + \bar{u},\id_d)]{\tilde{\eta}^i} \\
&= \exp\left( - \frac{d\log 2 + |\bar{u}|^2 + |\bar{u} + x_0|^2}{2} - \frac{|x_0|^2}{4} \right).
\end{align*}
And,
\begin{align*}
&\E{|\eta^i X_1^i|^2} = \frac{\E[X_1^i \sim \normal(x_0 + \bar{u},\id_d)]{(\tilde{\eta}^i)^2|X_1^i|^2}}{r_1^2} \\
&= \frac{\E[X_1^i \sim \normal(x_0 + \bar{u},\id_d)]{|X_1^i|^2\exp\left( - \left| X_1^i + \bar{u}\right|^2 \right)}}{r_1^2} \\
&= \frac{e^{-r_2}}{r_1^2(\sqrt{3})^d} \left( \left|\frac{\bar{u} - x_0}{3} \right|^2 + \frac{d}{3} \right) \\
& r_2 := \half\left( |x_0 + \bar{u}|^2 + 2|\bar{u}|^2 - \frac{|\bar{u} - x_0|^2}{3} \right)
\end{align*}

\noindent\textbf{\IPS:}
Recall that
\begin{align*}
\Y_1^i &= (\id_d + \Sigma_0^{(N)})\inv \Y_0^i + L_0^{(N)}(\zero - W_0^i) \\
L_0^{(N)} &= \Sigma^{(N)}_t(\Sigma^{(N)}_t + \id)\inv, \quad \Y_0^i \iid\normal(x_0,\id_d) \\
\hat{\Y}_1 &= \underbrace{(\id_d + \Sigma_0^{(N)})\inv}_{=:L_1} \underbrace{(\frac{1}{N}\sum_{i=1}^{N}\Y_0^i)}_{=:\tilde{\Y}} - \underbrace{L_0^{(N)}}_{=: L_2}( \underbrace{\frac{1}{N}\sum_{i=0}^{N} W_0^i}_{=:\tilde{W}})
\end{align*}
Then the optimal control is $(u_0^*)^{(N)}  := \hat{\Y}_1 - x_0$ and $\E{\hat{\Y}_1} = \frac{x_0}{2}$, so we have
\begin{align*}
&\var{u_0^{(N)}} = \E{\left|\hat{\Y}_1 - \frac{x_0}{2}\right|^2} \\
&= \E{\left|L_1\tilde{\Y} + L_2\tilde{W} - \frac{1}{2}x_0\right|^2} \\
&= \E{\left|L_1\tilde{\Y}\right|^2 + \left| L_2\tilde{W}\right|^2} + \frac{|x_0|^2}{4}
\end{align*}
since $\tilde{\Y}$ and $\tilde{W}$ are independent random variables. To obtain bounds on the first two terms, we combine the following result with the fact that $\tilde{\Y} \sim \normal (x_0, \frac{1}{N}\id_d), \tilde{W} \sim \normal (0, \frac{1}{N}\id_d)$.
\begin{lemma}
For any $S_1 \succeq 0$ and $x \in \R^d$,
\begin{align*}
|(\id_d + S_1)\inv x|^2 &\le |x|^2, \\
|S_1(\id_d + S_1)\inv x|^2 &\le |x|^2,
\end{align*}
\end{lemma}
\begin{proof}
It follows by expanding the expression, and writing $S_1$ using its spectral decomposition.
\end{proof}
To bound the third term, we use the following result.
\begin{lemma}
For any $S_1 \succeq 0$ and $x \in \R^d$,
\begin{align*}
|(\half\id_d - S_1)x|^2 \le \frac{5}{4}|x|^2.
\end{align*}
\end{lemma}
\hspace{12pt}\textit{Proof:}
First note that
\begin{align*}
\half\id_d - S_1(\id_d + S_1)\inv &= (\id_d + S_1)\inv - \half\id \\
((\id_d + S_1)\inv - \half\id_d)^2 &= (\id_d + S_1)^{-2} + \frac{1}{4}\id_d - (\id_d + S_1)\inv
\end{align*}
Hence,
\begin{align*}
|(\half\id_d - S_1)x|^2 &= x\tp ((\id_d + S_1)\inv - \half\id_d)^2 x \\
& \le x\tp((\id_d + S_1)^{-2} + \frac{1}{4}\id_d)x \le \frac{5}{4}|x|^2. \quad \rule{0.6em}{0.6em}
\end{align*}
This concludes the proof of Prop. \ref{prop:onestep}.

\end{document}